\documentclass[11pt]{article}
\usepackage{graphics}
\usepackage{graphicx}
\usepackage{url}
\usepackage{alltt}
\usepackage{a4wide}
\usepackage{amssymb}
\usepackage{amsmath}
\usepackage{latexsym}
\usepackage{xspace}
\usepackage{epsfig}
\usepackage{amssymb}
\usepackage{color}
\usepackage{appendix}
\usepackage{textcomp} 

\newcommand{\strS}{{\mathcal S}\xspace}

\newcommand{\ub}{\,\mbox{$\bullet$}\,}

\newcommand{\op}{\,\mbox{\bf\texttt{(}}\,}
\newcommand{\cp}{\,\mbox{\bf\texttt{)}}\,}

\chardef\other=12

\def\mmakeactive#1{\catcode`#1=\active\ignorespaces}

{
\mmakeactive\
\gdef\obeywhitespace{%
  \mmakeactive\^^M %
  \let^^M=\NewLine %
  \aftergroup\removebox %
  \obeyspaces %
}}

\def\NewLine{\par\indent}
\def\removebox{\setbox0=\lastbox}

\def\|{|}

\newtheorem{proposition}{Proposition}

\newtheorem{theorem}{Theorem}

\newtheorem{openproblem}{Open Problem}

\def\qed{\hfill\rule{2mm}{2mm}}

\newenvironment{proof}{\noindent{\bf Proof.}~}{\hfill\qed \vskip 5pt}

\begin{document}

\title{\bf Asymptotics of Canonical and Saturated RNA Secondary Structures}

\author{
Peter Clote\footnotemark[1]
\and
Evangelos Kranakis\footnotemark[2]
\and
Danny Krizanc\footnotemark[3]
\and
Bruno Salvy\footnotemark[4]
}

\def\qed{\hfill\rule{2mm}{2mm}}
\newcommand{\ep}{\qed}
\newcommand{\pf}{{\sc Proof }}
\newcommand{\pfo}{{\sc Proof (Outline) }}
\newcommand{\cA}{{\cal A}}
\newcommand{\cI}{{\cal I}}
\newcommand{\cP}{{\cal P}}
\newcommand{\cS}{{\cal S}}
\newcommand{\cC}{{\cal C}}
\newcommand{\cR}{{\cal R}}
\newcommand{\cG}{{\cal R}}
\newcommand{\cZ}{{\cal Z}}
\newcommand{\ex}{{\bf E}}

\maketitle

\def\thefootnote{\fnsymbol{footnote}}

\footnotetext[1]{Department of Biology,
Boston College, Chestnut Hill, MA 02467, USA.\hfill\break
Research partially supported by National Science Foundation
Grants DBI-0543506, DMS-0817971, and the RNA Ontology Consortium.
Additional support is gratefully acknowledged to
the Foundation Digiteo-Triangle
de la Physique and to Deutscher Akademischer Austauschdienst. {\tt clote@bc.edu}}
\footnotetext[2]{School of Computer Science, Carleton University,
K1S 5B6, Ottawa, Ontario, Canada.
Research supported in part by
Natural Sciences and Engineering Research Council
of Canada (NSERC) and Mathematics of Information Technology and Complex
Systems (MITACS). {\tt evankranakis@gmail.com}}
\footnotetext[3]{Department of Mathematics,
Wesleyan University, Middletown CT 06459, USA. {\tt dkrizanc@wesleyan.edu}}
\footnotetext[4]{Algorithms Project, Inria Paris-Rocquencourt, France. Supported in part by the Microsoft Research-Inria Joint Centre. {\tt Bruno.Salvy@inria.fr}}

\begin{abstract}
It is a classical result of Stein and Waterman 
that the asymptotic number of RNA secondary structures  is
$1.104366 \cdot n^{-3/2} \cdot 2.618034^n$.
In this paper, we study combinatorial asymptotics for two special subclasses
of RNA secondary structures -- {\em canonical} and {\em saturated} structures.
Canonical secondary structures are defined to
have no lonely (isolated) base pairs.
This class of secondary structures was introduced 
by Bompf{\"u}newerer et al., who noted
that the run time of Vienna RNA Package is substantially reduced
when restricting computations to canonical structures.
Here we provide an explanation for the speed-up, by 
proving that the asymptotic number of canonical RNA secondary structures
is $2.1614 \cdot n^{-3/2} \cdot 1.96798^n$ and that the 
expected number of base pairs in a canonical secondary structure
is $0.31724 \cdot n$. The asymptotic number of canonical secondary
structures was obtained much
earlier by Hofacker, Schuster and Stadler using a different method.

Saturated secondary structures have the property that
no base pairs can be added without violating the
definition of secondary structure (i.e. introducing a pseudoknot or
base triple).  Here we show that the asymptotic number
of saturated structures  is
$1.07427\cdot n^{-3/2} \cdot 2.35467^n$, the asymptotic expected number
of base pairs is $0.337361 \cdot n$, and
the asymptotic number of saturated stem-loop structures is
$0.323954 \cdot 1.69562^n$, in contrast to
the number $2^{n-2}$ of (arbitrary) stem-loop structures as 
classically computed by Stein and Waterman. 
Finally, we apply work of Drmota \cite{drmota1994,drmota} to
show that the density of states for [all resp. canonical resp. 
saturated] secondary structures is asymptotically Gaussian.  
We introduce
a stochastic greedy method to sample random saturated structures,
called {\em quasi-random saturated structures}, and show that
the expected number of base pairs of is $0.340633 \cdot n$.
\end{abstract}

\section{Introduction}

Imagine an undirected\footnote{We often describe the graph edges
of an undirected graph as $(i,j)$, where $i<j$, rather than
$\{ i,j \}$.}  graph, described by placing  graph 
vertices $1,\ldots,n$ along the periphery of a circle 
in a counter-clockwise manner, and placing graph edges as
chords within the circle. An {\em outerplanar} graph is a graph
whose circular representation is planar; i.e. there are
no crossings.
An RNA secondary structure, formally defined in Section \ref{section:DSV},
is an outerplanar graph (no pseudoknots)
with the property that no vertex is incident to more than
one edge (no base triples) and that for every chord between vertices
$i,j$, there exist at least $\theta=1$  many vertices that are not
incident to any edge (hairpin requirement).
RNA secondary structure is equivalently defined to be
a well-balanced parenthesis expression $s_1,\ldots,s_n$
with dots, where if nucleotide $i$ is 
unpaired then $s_i = \bullet$, while if there is a base pair
between nucleotides $i<j$ then $s_i = \op$ and $s_j = \cp$. This latter
representation is known as the {\em Vienna representation} or
{\em dot bracket notation} (dbn).

Formally, a well-balanced parenthesis expression
$w_1 \cdots w_n$ can be defined as follows.
If $\Sigma$ denotes a finite alphabet, and $\alpha \in \Sigma$,
and $w=w_1 \cdots w_n \in \Sigma^*$ is
an arbitrary {\em word}, or sequence of characters drawn from $\Sigma$,
then $|w|_{\alpha}$ designates the number of occurrences of $\alpha$ in $w$.
Letting $\Sigma = \{ \op,\cp \}$, a word $w = w_1 \cdots w_n  \in \Sigma^*$ is
{\em well-balanced} if for all $1 \leq i < n$,
$|w_1 \cdots w_i|_{\op} \geq |w_1 \cdots w_i|_{\cp}$ and
$|w_1 \cdots w_n|_{\op} = |w_1 \cdots w_n|_{\cp}$.  
Finally, when considering
RNA secondary structures, we consider instead the alphabet $\Sigma = \{ \op,
\cp, \ub\}$, but otherwise the definition of well-balanced expression
remains unchanged.
The number of well-balanced parenthesis expressions of length $n$ over
the alphabet $\Sigma = \{ \op,\cp \}$ is
known as the Catalan number $C_n$, while that over
the alphabet $\Sigma = \{ \op,\cp,\ub \}$ is known as the
Motzkin number $M_n$ \cite{motzkin}.  
Stein and Waterman \cite{steinWaterman} 
computed the number $S_n$ of well-balanced parenthesis expressions
in the alphabet $\Sigma=\{ \op,\cp,\ub \}$, where there exist at least
$\theta=1$ occurrences of $\ub$ between corresponding left and
right parentheses $\op$ respectively $\cp$. It follows that $S_n$ is exactly
the number of RNA secondary structures on $[1,n]$, where
there exist at least $\theta=1$ unpaired bases in every hairpin loop.

In this paper, we are interested in specific classes of secondary
structure: {\em canonical} and {\em saturated} structures.
A secondary structure is canonical \cite{Bompfunewerer.jmb08}
if it has no lonely (isolated) base pairs. A secondary structure is
saturated \cite{zuker:saturated}
if no base pairs can be added without violating the
notion of secondary structure, formally defined in 
Section \ref{section:DSV}.  In order to compute parameters like 
asymptotic value for number of structures,
expected number of base pairs, etc. throughout this paper, we adopt the 
model of Stein and Waterman \cite{steinWaterman}.
In this model, any position (nucleotide, also known as base) 
can pair with any other position, and 
every hairpin loop must contain at least $\theta=1$
unpaired bases; i.e. if $i,j$ are paired, then $j-i>\theta$.
This latter condition is due to steric constraints for RNA. At the risk of
additional effort, the combinatorial methods of this paper could be applied
to handle the situation of most
secondary structure software, which set $\theta=3$.

\subsection{Examples of secondary structure representations}

Figure~\ref{fig:rna1}
gives equivalent views of the secondary structure of 5S ribosomal
RNA with GenBank accession number NC\_000909 
of the methane-generating archaebacterium
{\em Methanocaldococcus jannaschii}, as determined by
comparative sequence analysis and taken from the {\em 5S 
Ribosomal RNA Database} \cite{Szymanski.nar00}
located at \url{http://rose.man.poznan.pl/5SData/}.
The sequence and its secondary structure in (Vienna) dot bracket 
notation are as follows:
\begin{quote}
\begin{tiny}
\verb?UGGUACGGCGGUCAUAGCGGGGGGGCCACACCCGAACCCAUCCCGAACUCGGAAGUUAAGCCCCCCAGCGAUGCCCCGAGUACUGCCAUCUGGCGGGAAAGGGGCGACGCCGCCGGCCAC?\\
\verb?((((.(((((((....(((((((......((((((.............))))..)).....))))).))...(((((.....(((((....)))))....)))))...))))))))))).?
\end{tiny}
\end{quote}

Equivalent representations for the same secondary structure may be 
produced by software {\tt jViz} \cite{Wiese.itn05}, as depicted in 
Figure \ref{fig:rna1}. The left panel of this figure depicts the
{\em circular Feynman diagram} (i.e. outerplanar graph representation),
the middle panel depicts the
{\em linear Feynman diagram}, and the right panel depicts
the {\em classical} representation. This latter representation,
most familiar to biologists, may also be obtained by {\tt RNAplot}
from the Vienna RNA Package \cite{hofacker:ViennaWebServer}.

\begin{figure}[htbp]
     \centering
\includegraphics[width=0.3\linewidth]{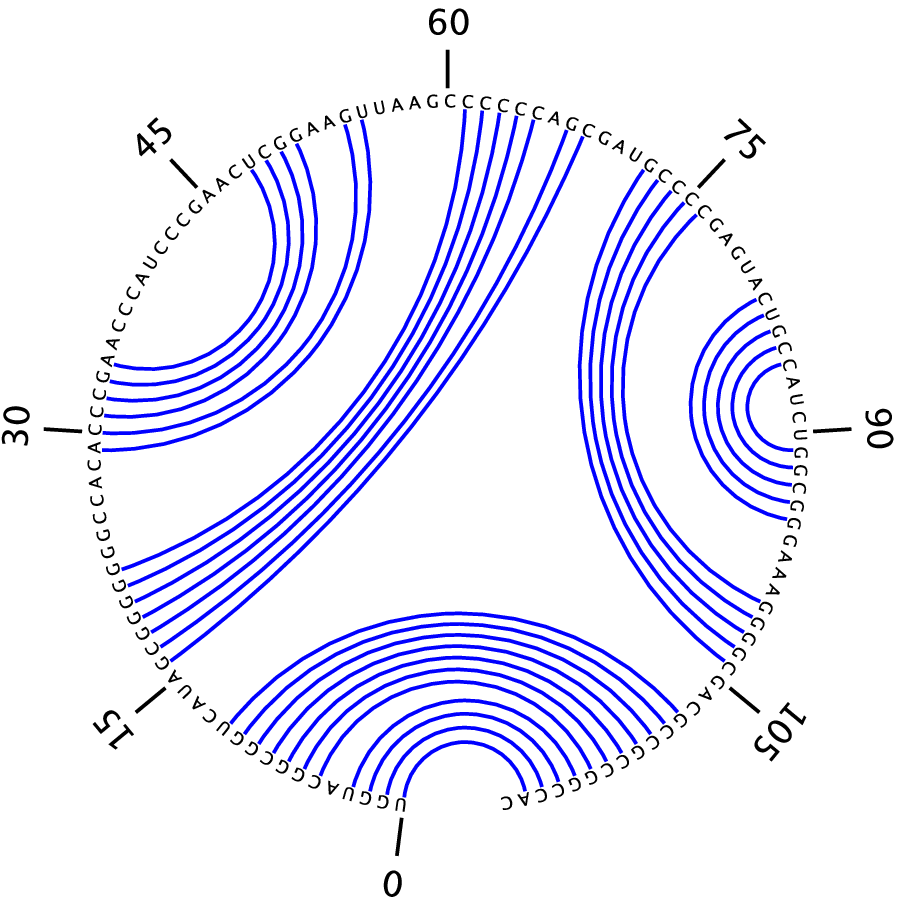}
\includegraphics[width=0.3\linewidth]{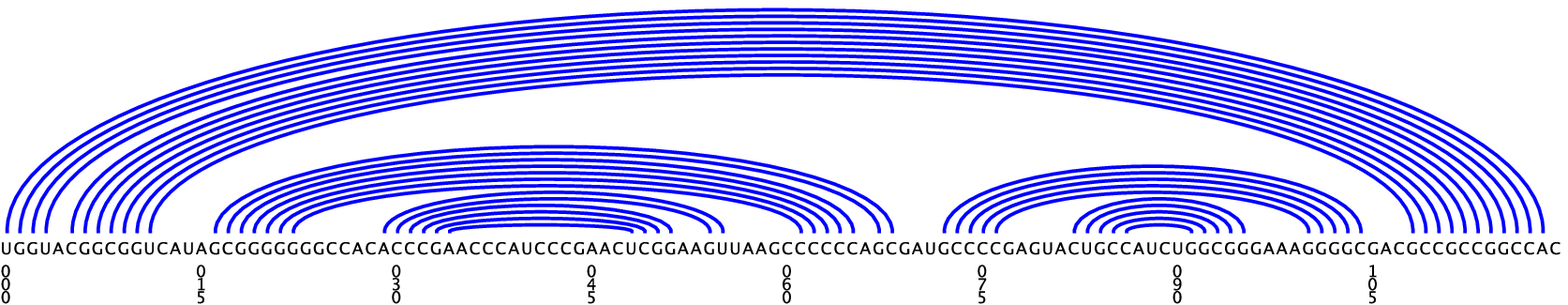}
\includegraphics[width=0.3\linewidth]{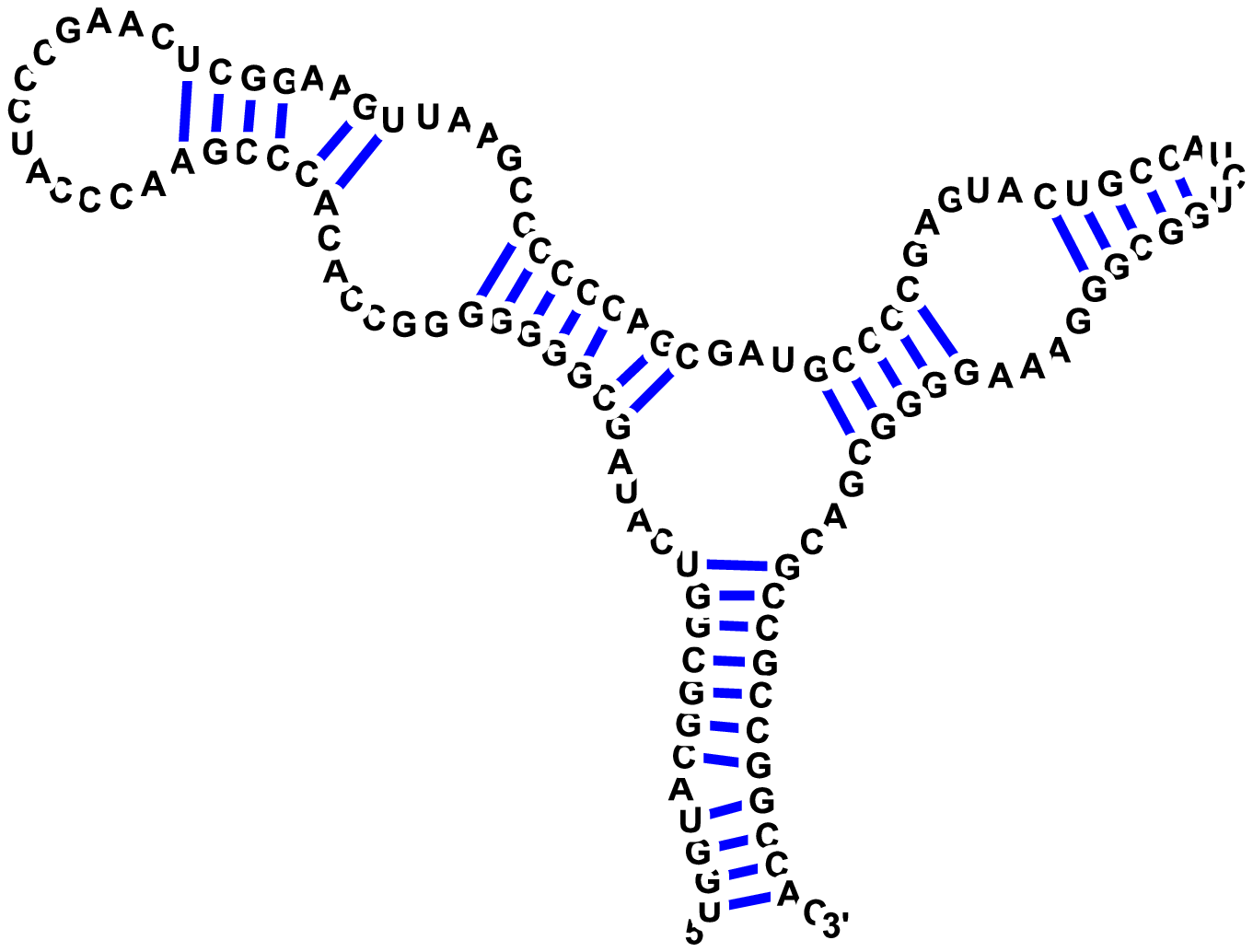}
     \caption{Depiction of 5S ribosomal RNA from {\em M. Jannaschii} with
     GenBank accession number NC\_000909. Equivalent representations as
    {\em (Left)} outerplanar graph (also called Feynman circular diagram),
    {\em (Middle)} Feynman linear diagram,
    {\em (Right)} classical diagram (most familiar to biologists).
     The sequence and
     secondary structure were taken from the 5S Ribosomal RNA Database
     \cite{Szymanski.nar00}, and the graph was created
     using {\tt jViz} \cite{Wiese.itn05}.}
\label{fig:rna1}
\end{figure}

\subsection{Outline and results of the paper}

In Section \ref{section:DSV}, 
we review a combinatorial method, known as the DSV methodology and
the important Flajolet-Odlyzko Theorem, which allows one to obtain
asymptotic values of Taylor coefficients of analytic generating functions 
$f(z) = \sum_{i=1}^{\infty} a_i z^i$ by
determining the dominant singularity of $f$.
The description of the DSV methodology and Flajolet-Odlyzko theorem
is not meant to be self-contained, although we very briefly describe the
broad outline. For a very clear review of
this method, with a number of example applications, please see
\cite{lorenzPontyClote:asymptotics} or the recent monograph of
Flajolet and Sedgewick \cite{flajoletBook}.

In Section \ref{section:numCanonStr}, we compute the asymptotic number 
$2.1614 \cdot n^{-3/2} \cdot 1.96798^n$ of canonical secondary structures,
obtaining the same value obtained by Hofacker, Schuster and Stadler
\cite{hofacker98} by a different method, known as the Bender-Meir-Moon
method.  In Section \ref{section:expBPcanonicalStr}
we compute the expected number 
$0.31724 \cdot n$ of base pairs in canonical secondary structures.
In Section \ref{section:numSatStr}, we apply the DSV methodology to
compute the asymptotic number $1.07427 \cdot n^{-3/2} \cdot 2.35467^n$
of saturated structures, while in Section
\ref{section:expNumBasePairsSatStr}, we compute the
expected number $0.337361  \cdot  n$ of base pairs of saturated structures.
In Section \ref{section:saturatedStemLoops}, we compute the 
asymptotic number $0.323954 \cdot 1.69562^n$ of saturated stem-loop structures,
which is substantially smaller than the number $2^{n-2}-1$ of (all) stem-loop
structures, as computed by Stein and Waterman \cite{steinWaterman}.

In Section \ref{qUniform distribution}, we consider a natural stochastic
process to generate random saturated structures, called in the sequel
{\em quasi-random saturated structures}.  The stochastic process 
adds base pairs, one at a time, according to the uniform distribution,
without violating any of the constraints of a structure. The main result
of this section is that asymptotically, the  expected number of
base pairs in quasi-random saturated structures
is $0.340633\cdot n$, rather close to the
expected number $0.337361  \cdot  n$ of base pairs of saturated structures.
The numerical proximity of these two values suggests that stochastic greedy
methods might find application in other areas of random graph theory.
In Section~\ref{conclusion} we provide some concluding remarks.

At the web site 
\url{http://bioinformatics.bc.edu/clotelab/SUPPLEMENTS/JBCBasymptotics/},
we have placed Python programs and Mathematica code used in computing
and checking the asymptotic number of canonical and saturated
secondary structures, as well as the Maple code for checking Drmota's
\cite{drmota} conditions to deduce the asymptotic normality of the 
density of states of RNA structures.

\section{DSV methodology}
\label{section:DSV}

In this section, we 
describe a combinatorial method sometimes called
\emph{DSV methodology}, after Delest, Sch\"utzenberger and
Viennot, which is a special case of what is called the \emph{symbolic method} in combinatorics, described at length in~\cite{flajoletBook}.
See also the Appendix of \cite{lorenzPontyClote:asymptotics} for a detailed
presentation of this method. This method enables one to obtain information on the number of combinatorial configurations defined by finite rules, for any size. This is done by translating those rules into equations satisfied by various \emph{generating functions}. A second step is to extract asymptotic expansions from these equations. This is done by studying the singularities of these generating functions viewed as analytic functions.

Since our goal is to derive asymptotic numbers of structures, following 
standard
convention we define an RNA secondary structure on a length $n$ sequence
to be a set of ordered pairs
$(i,j)$, such that $1 \leq i<j \leq n$ and the following are satisfied.
\begin{enumerate}
\item
{\em Nonexistence of pseudoknots:}
If $(i,j)$ and $(k,\ell)$ belong to $S$, then it is not the case that
$i<k<j<\ell$.
\item
{\em No base triples:}
If $(i,j)$ and $(i,k)$ belong to $S$, then $j=k$;
if $(i,j)$ and $(k,j)$ belong to $S$, then $i=k$.
\item
{\em Threshold requirement:}
If $(i,j)$ belongs to $S$, then $j-i > \theta$, where $\theta$, generally
taken to be equal to $3$, is the minimum number of unpaired bases in
a hairpin loop; i.e. there must be
at least $\theta$ unpaired bases in a hairpin loop.
\end{enumerate}
Note that the definition of secondary structure does not mention
nucleotide identity -- i.e.
we do {\em not} require base-paired positions $(i,j)$ to be occupied
by Watson-Crick or wobble pairs. For this reason, at times we may
say that $S$ is a secondary structure on $[1,n]$, rather than
saying that $S$ is a structure for RNA sequence of length $n$.
In particular, an expression such as ``the asymptotic number of structures
is $f(n)$'' means that the asymptotic number of structures on $[1,n]$ is
$f(n)$. 

\subsubsection*{Grammars}
We now proceed with basic definitions related to context-free grammars.
If $A$ is a finite
alphabet, then $A^*$ denotes the set of all finite sequences (called \emph{words}) of characters
drawn from $A$.
Let $\Sigma$ be the set consisting of the symbols for
left parenthesis $\op$, right parenthesis $\cp$, and dot $\bullet$, used
to represent a secondary structure in Vienna notation. 
A \emph{context-free} grammar (see, e.g., \cite{lewisPapadimitriou})
for RNA secondary structures is given by
$G = (V,\Sigma,\mathcal{R},S_0)$, where $V$ is a finite set of
nonterminal symbols 
(also called variables), $\Sigma = \{ \bullet, \op,\cp \}$,
$S_0 \in V$ is the {\em start} nonterminal, and
\[
\mathcal{R} \subseteq V \times (V \cup \Sigma)^*
\]
is a finite set of production rules. Elements of
$\mathcal{R}$ are usually denoted by $A \rightarrow w$, rather than
$(A,w)$. If rules
$A \rightarrow \alpha_1$,\ldots,
$A \rightarrow \alpha_m$ all have the same left-hand side, then
this is usually abbreviated by $A \rightarrow \alpha_1 \| \cdots \| \alpha_m$.

If $x,y \in (V \cup \Sigma)^*$ and $A \rightarrow w$ is
a rule, then by replacing the occurrence of $A$ in
$x A y$ we obtain $x w y$. Such a derivation in one
step is denoted by
$xAy \Rightarrow_G xwy$, while the
reflexive, transitive closure
of $\Rightarrow_G$ is denoted $\Rightarrow^*_G$.
The \emph{language} generated by context-free grammar $G$ is
denoted by $L(G)$, and defined by
\[
L(G) = \{ w \in \Sigma^* : S_0 \Rightarrow^*_G w \}.
\]
For any nonterminal $S \in V$, we also write $L(S)$ to denote
the language generated by rules from $G$ when using start symbol $S$.
A derivation of word $w$ from start symbol $S_0$ using grammar $G$
is a {\em leftmost} derivation, if each successive rule application is
applied to replace the leftmost nonterminal occurring in the intermediate
expression.
A context-free grammar $G$ is {\em non-ambiguous}, if there is
no word $w \in L(G)$ which admits two distinct leftmost derivations. 
This notion is important since it is only when applied to non-ambiguous grammars that the 
DSV methodology leads to exact counts.

For the sake of readers unfamiliar with context-free
grammars, we present some examples to illustrate the previous concepts.
Consider the following grammar $G$, which generates the collection of
well-balanced parenthesis strings, including the empty  
string.\footnote{A well-balanced parenthesis string is a word over~$\Sigma=\{(,)\}$ with as many closing parentheses as opening ones and such that when reading the word from left to right, the number of opening parentheses read is always at least as large as the number of closing parentheses.
RNA secondary structures can be considered to be well-balanced
parenthesis strings that also contain possible occurrences of  $\ub$,
and for which there exist at least $\theta$ occurrences of $\ub$ between
corresponding left and right parentheses $\op$ respectively $\cp$.}
Define $G = (V,\Sigma,R,S)$, where
the set $V$ of variables (also known as nonterminals) is
$\{ S\}$, the set $\Sigma$ of terminals is
$\{ \op,\cp \}$, where $S$ is the start symbol, and
where the set $R$ of rules is given by
\[
S \rightarrow \epsilon | \op S \cp | SS
\]
Here $\epsilon$ denotes the empty string.
We claim that $G$ is an ambiguous grammar. Indeed, consider the following
two leftmost derivations, where we denote the order of rule applications
$r1 := S \rightarrow \epsilon$, 
$r2 := S \rightarrow SS$, 
$r3 := S \rightarrow \op S \cp$, 
by placing the rule designator under the arrow. Clearly the leftmost derivation
\[
S \stackrel{\rightarrow}{\mbox{\tiny r2}}
SS \stackrel{\rightarrow}{\mbox{\tiny r2}}
SSS \stackrel{\rightarrow}{\mbox{\tiny r3,r1}}
\op \cp SS \stackrel{\rightarrow}{\mbox{\tiny r3,r1}}
\op \cp \op \cp S \stackrel{\rightarrow}{\mbox{\tiny r3,r1}}
\op \cp \op \cp \op \cp
\]
is distinct from the leftmost derivation
\[
S \stackrel{\rightarrow}{\mbox{\tiny r2}}
SS \stackrel{\rightarrow}{\mbox{\tiny r3,r1}}
\op  \cp S \stackrel{\rightarrow}{\mbox{\tiny r2}}
\op \cp \op S \cp S \stackrel{\rightarrow}{\mbox{\tiny r3,r1}}
\op \cp \op \cp S \stackrel{\rightarrow}{\mbox{\tiny r2}}
\op \cp \op \cp \op S \cp \stackrel{\rightarrow}{\mbox{\tiny r1}}
\op \cp \op \cp \op  \cp
\]
yet both generate the same well-balanced parenthesis string.
For the same reason, the grammar with rules
\[
S \rightarrow  \ub | \ub S | \op S \cp | SS
\]
generates precisely the collection of non-empty RNA secondary structures,
yet this grammar is ambiguous, and we would obtain an overcount by
applying the DSV methodology.  In contrast, the grammar whose rules are
\[
S \rightarrow  \ub |  \ub S | \op S \cp |  \op S \cp S
\]
is easily seen to be non-ambiguous and to generate all {\em non-empty}
RNA secondary structures.

\subsubsection*{Generating Functions}
Suppose that $G=(V,\Sigma,\mathcal{R},S)$ is a non-ambiguous
context-free grammar which generates a collection $L(S)$ of objects (e.g.
canonical secondary structures). To this grammar is associated a generating function
$S(z)= \sum_{n=0}^{\infty} s_n z^n$,
such that the $n$th Taylor coefficient $[z^n]S(z) = s_n$ represents the
number of objects we wish to count. In the sequel, $s_n$ will represent
the number of canonical secondary structures for RNA sequences of length $n$.
The DSV method uses Table~\ref{table:DSV} 
in order to translate
the grammar rules of $\mathcal{R}$ into a system of equations for the generating functions.

\begin{table}
\begin{center}
\begin{tabular}{|l|l|}
	\hline
	\mbox{Type of nonterminal} & \mbox{Equation for the g.f.} \\ \hline
	$S \to T \; | \; U$ & $S(z) = T(z) + U(z)$\\
		$S \to T\,U$        & $S(z) = T(z)U(z)$\\
		$S \to t$           & $S(z) = z$\\
		$S \to \varepsilon$ & $S(z) = 1$ \\ \hline
\end{tabular}
\end{center}
\caption{Translation between context-free grammars and generating functions.
Here, $G=(V,\Sigma,\mathcal{R},S_0)$ is a given context-free grammar,
$S$, $T$ and $U$ are any nonterminal symbols in $V$,  and
$t$ is a terminal symbol in $\Sigma$. The generating functions
for the languages $L(S)$, $L(T)$, $L(U)$ are respectively denoted by
$S(z)$, $T(z)$, $U(z)$.}
\label{table:DSV}
\end{table}

\subsubsection*{Asymptotics}
In the sequel, we often compute the asymptotic value of the Taylor coefficients
of generating functions by first applying the DSV methodology, then using
a simple corollary of a result of Flajolet  and Odlyzko~\cite{FlaOdl90}.
That corollary is restated here as the following theorem.
\begin{theorem}[Flajolet  and Odlyzko]
\label{thm:flajolet}
Assume that $S(z)$ has a singularity at $z=\rho>0$, is analytic in
the rest of the region $\triangle \backslash {1}$, depicted in
Figure \ref{fig:FlajoletTriangle},
and that as $z \rightarrow\rho$ in $\triangle$,
\begin{equation}\label{alg-sing}
S(z) \sim K(1-z/\rho)^{\alpha}.
\end{equation}
Then, as $n \rightarrow \infty$, if $\alpha \notin {0, 1, 2, ...}$,
\begin{eqnarray*}
 s_n \sim \frac{K}{\Gamma(-\alpha)} \cdot n^{-\alpha-1}\rho^{-n}.
\end{eqnarray*}
\end{theorem}
It is a consequence of Table~\ref{table:DSV} that the generating series of context-free grammars are algebraic (this is the celebrated theorem of Chomsky and Sch\"utzenberger~\cite{ChomskySchutzenberger1963}). In particular this implies that they have positive radius of convergence, a finite number of singularities, and their behaviour in the neighborhood of their singularities is of the type~\eqref{alg-sing}. (See~\cite[\S VII.6--9]{flajoletBook} for an extensive treatment.)

A singularity of minimal modulus as in Theorem~\ref{thm:flajolet} is called a \emph{dominant singularity}. The location of the dominant singularity may be a source of difficulty. The simple case is when an explicit expression is obtained for the generating functions; this happens for canonical secondary structures. The situation when only the system of polynomial equations is available is more involved; we show how to deal with it in the case of saturated structures. 

\begin{figure}[htbp]
     \centering
\scalebox{0.55}{\input{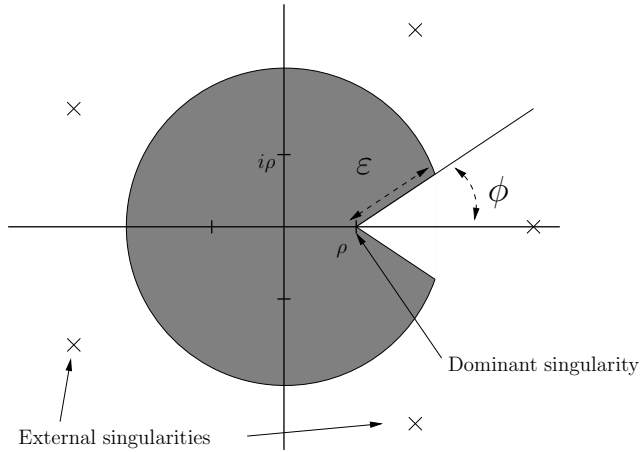}}
\caption{The shaded region $\triangle$ where, except at $z=\rho$, the generating function $S(z)$ must be analytic. 
} 
\label{fig:FlajoletTriangle}
\end{figure}

\subsection{Asymptotic number of canonical secondary structures}
\label{section:numCanonStr}

In  Bompf\"unewerer et al.  \cite{Bompfunewerer.jmb08},
the notion of {\em canonical secondary structure} $S$ is defined as a
secondary structure having no {\em lonely} (isolated) base pairs; i.e. formally,
there are no base pairs $(i,j) \in S$ for which both $(i-1,j+1)\not\in S$ 
and $(i+1,j-1)\not\in S$. In this section, we compute the asymptotic
number of canonical secondary structures.
Throughout this section, secondary structure is interpreted to mean a
secondary structure on an RNA sequence of length $n$, for which each
base can pair with any other base (not simply Watson-Crick and wobble pairs),
and  with minimum number $\theta$ of
unpaired bases in every hairpin loop set to be $1$. At the cost of
working with more complex expressions, by the same method, one could analyze
the case when $\theta=3$, which is assumed for the software {\tt mfold}
\cite{zuker:mfoldWebserver} and {\tt RNAfold} \cite{hofacker:ViennaWebServer}.

\subsubsection*{Grammar}
Consider the context-free grammar $G=(V,\Sigma,{\cal R},S)$,
where $V$ consists of nonterminals $S,R$, $\Sigma$ consists of
the terminals $\ub, \op, \cp$, $S$ is the
start symbol and $\cal R$ consists of the following rules:
\begin{eqnarray}
\label{grammar:canonicalStr}
S &\rightarrow &\bullet | S \bullet | \op R \cp | S \op R \cp\\
R &\rightarrow &\op \bullet \cp |  \op R \cp | \op S \op R \cp \cp |
\op S \bullet \cp \nonumber
\end{eqnarray}
The nonterminal $S$ is intended to generate all {\em nonempty canonical}
secondary structures. In contrast, the nonterminal $R$ is intended to
generate all secondary structures which become canonical when surrounded
by a closing set of parentheses.
We prove by induction on expression length that the grammar
$G$ is non-ambiguous and generates all nonempty canonical secondary structures.

Define context-free grammar $G_R$ to consist of the collection 
$\cal R$ of rules from $G$, defined above, 
with starting nonterminal $S$, respectively. Formally,
\begin{eqnarray*}
G_R &= & (V,\Sigma,{\cal R},R).
\end{eqnarray*}
Let $L(G)$, $L(G_R)$ denote the languages generated
respectively by grammars $G,G_R$.
Now define languages $L_1,L_2$ of {\em nonempty}
secondary structures with $\theta=1$  by
\begin{eqnarray*}
L_1 &=& \{ S : \mbox{$S$ is canonical}\}\\
L_2 &=& \{ S : \mbox{$\op S \cp$ is canonical}\}.
\end{eqnarray*}
Note that structures like $\ub \ub \op \ub \cp$ and
$\op \ub \cp \op \ub \cp$ belong to $L_1$, but not to $L_2$,
while structures like $\op \op \ub \cp \cp$ belong to both $L_1,L_2$.
Note that any structure $S$ belonging to $L_2$ must be of the form
$\op S_0 \cp$; indeed, if $S$ were not of this form, but rather of the
form either $\ub S_0$ or $\op S_0 \cp S_1$, then by $\op S \cp$ would
have an outermost lonely pair of parentheses.
\medskip

\noindent
{\sc Claim.} $L_1=L(G)$, $L_2 = L(G_R)$. \hfill\break
\medskip

\noindent
{\sc Proof of Claim.} 
Clearly $L_1 \supseteq L(G)$, $L_2  \supseteq  L(G_R)$, so we show the
reverse inclusions by induction; i.e.
by induction on $n$, we prove that
$L_1 \cap \Sigma^n \subseteq L(G) \cap \Sigma^n$, 
$L_2 \cap \Sigma^n \subseteq L(G_R) \cap \Sigma^n$.
\medskip

\noindent
{\sc Base case:} $n=1$. 
Clearly
$L(G) \cap \Sigma = \{ \ub \} = L_1 \cap \Sigma$,
$L(G_R) \cap \Sigma = \emptyset  = L_2 \cap \Sigma$.
\medskip

\noindent
{\sc Induction case:} Assume that the claim holds for all $n<k$.
\medskip

\noindent
{\em Subcase 1.} Let
$\strS$ be a canonical secondary structure with length $|\strS|=k>1$.
Then either 
{\em (1)} $\strS = \ub \strS_0$, where $\strS_0 \in L_1$, or
{\em (2)} $\strS = \op \strS_0 \cp$, where $\strS_0 \in L_2$, or
{\em (3)} $\strS = \op \strS_0 \cp \strS_1$, where $\strS_0 \in L_2$ and
$\strS_1 \in L_1$.  Each of these cases corresponds to a different rule
having left side $S$, hence by the induction hypothesis, it follows that
$\strS \in L(G)$.

\noindent
{\em Subcase 2.} Let
$\strS \in L_2$ be a secondary structure with length $|S|=k>1$,
for which $\op \strS \cp$ is canonical.  If $\strS$ were of the form
$\ub \strS_0$ or $\op \strS_0 \cp \strS_1$, then $\op \strS \cp$ 
would not be canonical,
since its  outermost parenthesis pair would be a lonely pair. 
Thus $\strS$ is of the form $\op \strS_0 \cp$, where either
{\em (1)} $\strS_0$ begins with  $\ub$,  or
{\em (2)} $\strS_0$ is of the form $\op \strS_1 \cp$, where $\strS_1$ is
not canonical, but $\op \strS_1 \cp$ becomes canonical, or
{\em (3)} $\strS_0$ is of the form $\op \strS_1 \cp$, where $\strS_1$ is
canonical and $\op \strS_1 \cp$ is canonical as well.

In case {\em (1)}, $\strS_0$ is either $\ub$ or $\ub \strS_1$,
where $\strS_1$ is canonical. In case {\em (2)},
$\strS_0$ is of the form $\op \strS_1 \cp$, where $\strS_1$ must have
the property that $\op \strS_1 \cp$ is canonical. In case {\em (3)},
$\strS_0$ is of the form $\op \strS_1 \cp \strS_2$, where it must be that
$\op \strS_1 \cp$ is canonical and $\strS_2$ is canonical. By applying
corresponding rules and the induction hypothesis, it follows that 
$S \in L(G_R)$.

It now follows by induction that $L_1=L(G)$, $L_2 = L(G_R)$. A similar
proof by induction shows that the grammar $G$ is non-ambiguous.

\subsubsection*{Generating Functions}

Now, let $s_n$ denote the number of canonical secondary structures
on a length $n$ RNA sequence.  Then $s_n$ is the
$n$th Taylor coefficient of the generating function
$S(z) = \sum_{n \geq 0} s_n z^n$, denoted by
$s_n = [z^n]S(z)$.  Similarly, let
$R(z) = \sum_{n \geq 0} R_n z^n$ be the generating function
for the number of secondary structures on $[1,n]$
with $\theta=1$, which become canonical when surrounded by a closing
set of parentheses.

By Table~\ref{table:DSV}, the non-ambiguous grammar
(\ref{grammar:canonicalStr}) gives the following equations
\begin{eqnarray}
\label{eqn:canonicalStr1}
S(z) &=& z+S(z)z + R(z) z^2 + S(z)R(z) z^2 \\
\label{eqn:canonicalStr2}
R(z) &=& z^3 + R(z)z^2 + S(z)R(z)z^4 + S(z)z^3 
\end{eqnarray}
which can be solved explicitly (solve the
second equation for~$R$ and inject this in the first equation):
\begin{eqnarray}
\label{eqn:1}
S(z)=\frac{1-z-z^2+z^3-z^5-\sqrt{F(z)}}{2z^4}
\end{eqnarray}
and 
\begin{eqnarray}
\label{eqn:2}
S(z)=\frac{1-z-z^2+z^3-z^5+\sqrt{F(z)}}{2z^4}
\end{eqnarray}
where 
\begin{equation}\label{eqn:Fofz}
F(z) =4 z^5 \left(-1+z^2-z^4\right)+\left(-1+z+z^2-z^3+z^5\right)^2.
\end{equation}
When evaluated at $z=0$, Equation (\ref{eqn:2}) gives~$\lim_{r\rightarrow 0}S(z)=\infty$.
Since~$S(z)$ is known to be analytic at~0,
we conclude that $S(z)$ is given by~\eqref{eqn:1}.

\subsubsection*{Location of the dominant singularity}
The square root function $\sqrt{z}$ has a singularity at $z=0$,
so we are led to investigate the roots of $F(z)$.
A numerical computation with Mathematica\texttrademark ~gives the 10 roots
$0.508136$,
$4.11674$,
$-0.868214-0.619448 i$,
$-0.868214+0.619448 i$,
$-0.799805-0.367046 i$,
$-0.799805+0.367046 i$,
$0.410134-0.564104 i$,
$0.410134+0.564104 i$,
$0.945448-0.470929 i$,
$0.945448+0.470929 i$.
It follows that 
$\rho=0.508136$ is the root of $F(z)$ having smallest (complex) modulus.

\subsubsection*{Asymptotics}
Let $T(z)=\frac{ 1-z-z^2+z^3-z^5}{2z^4}$ and
factor $1-z/\rho$ out of $F(z)$ to obtain $Q(z)(1-z/\rho)  = F(z)$.
It follows that
\[
S(z)-T(\rho) = \frac{\sqrt{Q(\rho)}}{2\rho^4} \cdot (1-z/\rho)^{\alpha}+O(1-z/\rho),\qquad z\rightarrow\rho,
\]
where $\alpha = 1/2$.  This shows that $\rho$ is indeed a dominant singularity for~$S$. Note that for each~$n\ge1$, $S(z)$ and $S(z)-T(\rho)$ have the same Taylor coefficient of index~$n$, namely $s_n$.
Now, it is a direct consequence of Theorem~\ref{thm:flajolet} that
\begin{eqnarray}
\label{eqn:derivationNumCanonStr}
s_n \sim  \frac{K(\rho)}{\Gamma(-\alpha)} \cdot n^{-\alpha-1} \cdot
(1/\rho)^n,\qquad n\rightarrow\infty
\end{eqnarray}
where $\alpha=1/2$ and $K(z) = \frac{\sqrt{Q(z)}}{2z^4}$.
Plugging $\rho= 0.508136$ into equation (\ref{eqn:derivationNumCanonStr}),
we derive the following theorem, first obtained by Hofacker, Schuster and
Stadler \cite{hofacker98} by a different method.
\begin{theorem}
The asymptotic number of canonical secondary structures on $[1,n]$ is 
\begin{equation}
\label{eqn:asymNumCanonStr}
2.1614 \cdot n^{-3/2} \cdot 1.96798^n.
\end{equation}
\end{theorem}

\subsection{Asymptotic expected number of base pairs in canonical structures}
\label{section:expBPcanonicalStr}

In this section, we derive the expected number of base pairs in 
canonical secondary structures on $[1,n]$.

\subsubsection*{Generating Functions}
The DSV methodology is actually able to produce \emph{multivariate} generating series.
Modifying the equations (\ref{eqn:canonicalStr1},\ref{eqn:canonicalStr2}) 
by adding a new
variable $u$, intended to count the number of base pairs, we get
\begin{eqnarray}
\label{eqn:expBPcanonicalStr1}
S(z,u) &=& z+S(z,u)z + R(z,u)u z^2 + S(z,u)R(z,u)u z^2 \\
\label{eqn:expBPcanonicalStr2}
R(z,u) &=& u z^3 + R(z,u)uz^2 + S(z,u)R(z,u)u^2z^4 + S(z,u)uz^3 .
\end{eqnarray}
This can be solved as before 
to yield
the solution\footnote{Since $S(z,u)$ is known to be
analytic at $0$, 
we have discarded one of the two solutions as before.}
\begin{eqnarray*}
S(z,u) &=& \sum_{n\ge 0} \sum_{k\ge 0} s_{n,k} z^n u^k\\
 & = & 2 u^2 z^4 \left(1-z-u z^2+u z^3-u^2 z^5 -  \right. \\
 & & \left.
\sqrt{4 u^2 z^5 \left(-1+u z^2-u^2 z^4\right)+
\left(-1+z+u z^2-u z^3+u^2 z^5\right)^2}\right)
\end{eqnarray*}
Here, the coefficient~$s_{n,k}$ is the number of canonical secondary structures of size~$n$ with~$k$ base pairs. 
Using a classical observation on multivariate generating functions, we recover
the expected
number of base pairs in a canonical secondary structure 
on $[1,n]$ using the partial derivative
of $S(z,u)$; indeed,
\begin{eqnarray*}
\frac{[z^n]\frac{\partial S(z,u)}{\partial u}(z,1)}{[z^n]S(z,1)} & = &
\frac{[z^n]\left(\sum_{i\ge 0} \sum_{k\ge 0} s_{i,k} z^i k  u^{k-1}\right)(z,1)}{s_{n}}\\
& = & \frac{\sum_{k\ge 0} s_{n,k} k }{s_{n}}
 =  \sum_{k\ge 0}  k \frac{s_{n,k}}{s_n},
\end{eqnarray*}
and ${s_{n,k}}/{s_{n}}$ is the (uniform) probability
that a canonical secondary structure on $[1,n]$ has
exactly $k$ base pairs.

We compute that
$G(z) = \frac{\partial S(z,u)}{\partial u}(z,1)$ satisfies
\[
G(z) = \frac{-(z^2 - 2) (T(z) - \sqrt{F(z)} + z \sqrt{F(z)})}
{2z^4 \sqrt{F(z)}}
\]
where $T(z) = (1 - 2z + 2z^3 - z^4 - 3z^5 + z^6)$ and $F(z)$ is as in~\eqref{eqn:Fofz}.
Simplification yields
\begin{eqnarray*}
G(z) &= &\frac{-(z^2 - 2)(z-1)}{2z^4} - \frac{T(z)(z^2-2)}{2z^4}
\cdot \left( \frac{1}{\sqrt{F(z)}} \right). \\
\end{eqnarray*}
\subsubsection*{Asymptotics}
From this expression, it is clear that the dominant singularity is again located at the same~$\rho = 0.508136$.
A local expansion there gives
\[G(z)\sim K(\rho)(1-z/\rho)^{-1/2},\qquad z\rightarrow\rho\]
with $K(z) = - \frac{Q(z)^{-1/2}T(z)(z^2-2)}{2z^4}$.
By Theorem~\ref{thm:flajolet},
we obtain the asymptotic value
\begin{eqnarray}
\label{eqn:flajoletExpBPcanon}
\frac{K(\rho)}{\Gamma(-\alpha)} \cdot n^{-3/2} \cdot
(1/\rho)^n.
\end{eqnarray}
Plugging $\rho= 0.508136$ into equation 
(\ref{eqn:flajoletExpBPcanon}), we find the asymptotic value of
$[z^n]\frac{\partial S(z,u)}{\partial u}(z,1)$ is
\begin{eqnarray}
\label{eqn:asympPartial}
0.68568 \cdot n^{-1/2} \cdot 1.96798^n.
\end{eqnarray}
Dividing (\ref{eqn:asympPartial})
by the asymptotic number $[z^n]S(z)$ of canonical secondary structures,
given in (\ref{eqn:asymNumCanonStr}), we have the following theorem.
\begin{theorem}
\label{thm:expNumBasePairsCanonStr}
The asymptotic expected 
number of base pairs in canonical secondary structures is $0.31724 \cdot n$.
\end{theorem}

\subsection{Asymptotic number of saturated structures}
\label{section:numSatStr}

An RNA secondary structure is {\em saturated} if 
it is not possible to add any base pair without
violating the definition of secondary structures. If one models
the folding of an RNA secondary structure as a random walk on a Markov
chain (i.e. by the Metropolis-Hastings algorithm), then saturated 
structures correspond to {\em kinetic traps} with respect to the
Nussinov energy model \cite{nussinovJacobson}.
The asymptotic number of saturated structures 
was determined in \cite{Clote.jcb06} by using a method known
as Bender's Theorem, as rectified by Meir and Moon \cite{meirMoon}.
In this section, we apply the DSV methodology to obtain the same
asymptotic limit, and in the next section we obtain the expected
number of base pairs of saturated structures.

\subsubsection*{Grammar}

Consider the context-free grammar with nonterminal symbols $S,R$,
terminal symbols $\bullet, \op,\cp$, start symbol $S$ and production rules
\begin{eqnarray}
\label{eqn:CFGforSatStrS}
S & \rightarrow & \bullet | \bullet \bullet | R \bullet | R \bullet \bullet | \op S \cp |  S \op S \cp \\
\label{eqn:CFGforSatStrR}
R & \rightarrow & \op S \cp |  R \op S \cp 
\end{eqnarray}
It can be shown by induction on expression length that
$L(S)$ is the set of saturated structures, and $L(R)$ is the
set of saturated structures with no {\em visible} position; i.e.
external to every base pair \cite{Clote.jcb06}.
Here, position $i$ is visible in a secondary structure $T$
if it is external to every base pair of $T$; i.e. for all $(x,y) \in T$,
$i<x$ or $i>y$. 

\subsubsection*{Generating Functions}
Let
\begin{eqnarray}
\label{eqn:defGenFunS}
S(z) = \sum_{i=0}^{\infty} s_i \cdot z^i, \qquad
R(z) = \sum_{i=0}^{\infty} r_i \cdot z^i
\end{eqnarray}
denote the generating functions $S$ resp. $R$,
corresponding to the problems of counting
number of saturated secondary structures resp. number of saturated structures
having no visible positions.  Applying Table~\ref{table:DSV},
we are led to the equations
\begin{eqnarray}
\label{eqn:grammarRulesS}
S &=& z+ z^2 + z R + z^2 R + z^2 S + z^2 S^2\\
\label{eqn:grammarRulesR}
R &=& z^2 S + z^2 R S.
\end{eqnarray}

\subsubsection*{Location of the dominant singularity}
By first solving~\eqref{eqn:grammarRulesR} for~$R$ and injecting in~\eqref{eqn:grammarRulesS}, we get
\begin{equation}\label{eqn:S}
	S=z+z^2+z^2S+z^2S^2+(z+z^2)\frac{z^2S}{1-z^2S},
\end{equation}
which upon normalizing gives a polynomial equation of the third degree
\begin{equation}
\label{eqn:resultantP}
P(z,S) = -S^3 z^4+z (1+z)-S^2 z^2 \left(-2+z^2\right)+S \left(-1+z^2\right) 
= 0.
\end{equation}
Unlike earlier work in this paper, direct solution of this
equation by Cardano's formulas gives expressions that are difficult to handle. 
Instead, we locate the singularity by appealing to general techniques for implicit generating functions~\cite[\S VII.4]{flajoletBook}.

%
%
%

By the implicit function theorem, singularities of $P(z,S)$ 
only occur when both $P$ and its partial derivative 
\begin{eqnarray}
\label{eqn:P_S}
\frac{\partial P}{\partial S}(z,S)
    = -1 + (1 + 4 S) z^2 - S (2 + 3 S) z^4
\end{eqnarray}
vanish simultaneously. 

The common roots of $P$ and $\partial P/\partial S$ can be located by eliminating~$S$ between those two equations, for instance using the classical theory of \emph{resultants} (see, e.g., \cite{lang:Algebra}). This gives a polynomial
%
\begin{equation}
Q(z) = 
z^{11} (1 + z) (4 + z - 7 z^2 - 28 z^3 - 32 z^4 + 4 z^6),
\end{equation}
that vanishes at all~$z$ such that~$(z,S)$ is a common root of~$P$ and~$\partial P/\partial S$.

Numerical computation of the roots of $Q$ yields
$0$, $-1$, $-2.29493$, $-0.854537$, $-0.244657 - 0.5601 i$,
$-0.244657 + 0.5601 i$, $0.424687$, $3.2141$.      

A subtle difficulty now lies in selecting among those points the dominant singularity of the analytic continuation of the solution~$S$ of~\eqref{eqn:S} corresponding to the combinatorial problem. Indeed, it is possible that one solution of~\eqref{eqn:S} is singular at a given~$r$ without the solution of interest being singular there. Considering such a singularity would result in an asymptotic expansion that is wrong by an exponential factor. One way to select the correct singularity is to apply a result by Meir and Moon~\cite{MeirMoon1989} to Equation~\eqref{eqn:S}. This results in a variant of the computation in~\cite{Clote.jcb06}.

Instead, we use Pringsheim's theorem (see, e.g., \cite{flajoletBook}).
\begin{theorem}[Pringsheim] If $S(z)$ has a series expansion at~0 that has nonnegative coefficients and a radius of convergence~$R$, then the point $z=R$ is a singularity of~$S(z)$.
\end{theorem}
In our example, there are only two possible real positive singularities, $0.424687$ and $3.2141$. The latter cannot be dominant, since it would lead to asymptotics of the form~$3.2141^{-n}$, i.e., an exponentially decreasing number of structures. Thus the dominant singularity is at $\rho=0.424687$. Since the moduli of the non-real roots of~$Q$ is $0.611203>\rho$, the conditions of Theorem~\ref{thm:flajolet} hold, provided the function behaves as required as~$z\rightarrow\rho$.

\subsubsection*{Asymptotics}

%
We now compute the local expansion of~$S(z)$ at~$\rho$.
From equation (\ref{eqn:P_S}), we have that
\begin{equation}
\label{eqn:PrhoS}
P(\rho,S) = 0.605047- 0.819641 S + 0.328189 S^2 - 0.0325295 S^3
\end{equation}
whose (numerical approximations of) roots are the double root
$S= 1.6569$ and single root $S= 6.77518$. It is easily checked that
$1.6569$ is the only root of equation (\ref{eqn:PrhoS})
in which $P(\rho,S)$ is increasing; thus we let $T = 1.6569$.

Recall Taylor's theorem in two variables
\[
f(x,y)=
\sum_{n=0}^{\infty}\sum_{k=0}^{\infty}
\frac{\partial^{n+k}f(x_0,y_0)}
{\partial x^{n}\partial y^{k}} \cdot
\frac{(x-x_0)^{n}}{n!} \cdot\frac{(y-y_0)^{k}}{k!}.
\]
We now expand $P(z,S)$ at $z=\rho$ and $S=T$ and invert this expansion. This yields
\begin{equation}
P(z,S)
= P(\rho,T) + \frac{\partial P}{\partial S}(\rho,T)(S-T)
+ \frac{\partial P}{\partial z}(\rho,T)(z-\rho)
+\frac{1}{2} \frac{\partial^2 P}{\partial S^2}(\rho,T)(S-T)^2 + \cdots
\end{equation}
where the dots indicate terms of higher order.
The first two terms are $0$, so  by denoting
$P_z = \frac{\partial P}{\partial z}(\rho,T)$ and
$P_{SS}= \frac{\partial^2 P}{\partial S^2}(\rho,T)$, we have
\begin{equation}
0=P = P_z(z-\rho) +\frac{1}{2} P_{zz}(S-T)^2 + 
O(S-T)^3 + O((z-\rho)(S-T)^2) + O((z-\rho)^2).
\end{equation}
Isolating~$(S-T)^2$ we get
\begin{eqnarray*}
(S-T)^2 &=& \frac{-2P_z(z-\rho)}{P_{SS}} + O((z-\rho)^2) +O((S-T)^3)\\
S-T 
 &=& \pm \sqrt{\frac{2\rho P_z}{P_{SS}}} \cdot \sqrt{1-z/\rho} + O(z-\rho).
\end{eqnarray*}
%
%
Since $[z^n]S(z)$ is the number of saturated secondary 
structures on $[1,n]$ and the Taylor coefficients in the expansion of
$\sqrt{1-z/\rho}$ are negative, we discard the positive root and thus obtain
\begin{equation}\label{eqn:expansionS}
S-T = -\sqrt{\frac{2\rho P_z}{P_{SS}}} \cdot \sqrt{1-z/\rho} + O(z-\rho).
\end{equation}

We now make use of Theorem \ref{thm:flajolet} as before
and recover the following result, proved earlier in
\cite{Clote.jcb06} by the Bender-Meir-Moon method.
\begin{theorem}
\label{thm:numSatStr}
The asymptotic number of saturated structures is
$1.07427 \cdot n^{-3/2} \cdot 2.35468^n$.
\end{theorem}

\subsection{Expected number of base pairs of saturated structures}
\label{section:expNumBasePairsSatStr}

In this section, we compute the expected number of base pairs of saturated 
structures, proceeding as in Section~\ref{section:expBPcanonicalStr} by first modifying the equations to obtain bivariate generating functions and then differentiating with respect to the new variable and evaluating at~1 to obtain the asymptotic expectation.


\subsubsection*{Generating Functions}
We first modify
equations (\ref{eqn:grammarRulesS},\ref{eqn:grammarRulesR}) by 
introducing the auxiliary variable $u$, responsible for
counting the number of base pairs:
\begin{eqnarray}
\label{eqn:grammarRulesSU}
S &=& z+ z^2 + z R + z^2 R + uz^2 S + uz^2 S^2\\
\label{eqn:grammarRulesRU}
R &=& u z^2 S + u z^2 R S.
\end{eqnarray}

Solving the second equation for~$R$ and injecting into the first one gives
\begin{equation}\label{eq:PzuS}
P(z,u,S) = 
 S u z^2 (z + z^2) - (-1 + S u z^2) (-S + z + z^2 + S u z^2 + 
    S^2 u z^2).
\end{equation}                                            

\subsubsection*{Asymptotics}
We are interested in the coefficients of~$\partial S/\partial u$ at~$u=1$. Differentiating~\eqref{eq:PzuS} with respect to~$u$ gives
%
%
\[
\frac{\partial P}{\partial u} +
\frac{\partial P}{\partial S} 
\frac{\partial S}{\partial u}  = 0.
\]
Using equation (\ref{eqn:expansionS}), we replace $S(z,1)$ 
by $T+K\sqrt{1-z/\rho} + O(1-z/\rho)$ in this equation to obtain
\begin{multline*}
\left(\rho^2T(1+2(1-\rho^2)T-2\rho^2T^2)+O(\sqrt{1-z/\rho})\right)+\\
\left((4 K \rho^2 - 2 K \rho^4 - 6 K \rho^4 T) \sqrt{1 - z/\rho} +
O(1-z/\rho)\right)\left.\frac{\partial S}{\partial u}\right|_{u=1}=0
\end{multline*}
and finally

\begin{eqnarray}
\label{eqn:numeratorExpectedNumBasePairsInSatStr}
\frac{\partial S}{\partial u}\left(z,1 \right) 
&\sim& -\frac{0.642305}{\sqrt{1 - z/\rho}}.
\nonumber
\end{eqnarray}
Applying Theorem~\ref{thm:flajolet} to
equation (\ref{eqn:numeratorExpectedNumBasePairsInSatStr}) gives
\begin{eqnarray*}
\rho^n[z^n] \frac{\partial S}{\partial u}(z,1) 
\sim \frac{0.642305}{\Gamma(1/2)} \cdot n^{-1/2} 
= 0.362417 \cdot n^{-1/2}.
\end{eqnarray*}
It follows that the
asymptotic expected number of base pairs in saturated structures on
$[1,n]$ is
\begin{eqnarray*}
\frac{[z^n]\frac{\partial S(z,u)}{\partial u}(z,1)}{[z^n]S(z,1)}  
\sim \frac{ 0.362417 \cdot  n^{-1/2} \cdot \rho^{-n} }
{1.07427 \cdot n^{-3/2} \cdot \rho^{-n}} 
= 0.337361  \cdot  n
\end{eqnarray*}
We have just proved the following.
\begin{theorem}
\label{thm:expectedNumBasePairsInSatStr}
The asymptotic expected number of base pairs for saturated structures
is $0.337361  \cdot  n$.
\end{theorem}
Since the Taylor coefficient
$s_{n,k}$ of generating function $S(z,u) = \sum_{n,k} s_{n,k}z^n u^k$
is equal to the number of saturated structures having $k$ base pairs,
it is possible that the methods of this section will suffice to solve
the following open problem.
\begin{openproblem}
Clearly, the maximum number of base pairs in a saturated structure on
$[1,n]$ where $\theta=1$ is $\lfloor \frac{n-1}{2} \rfloor$.
For fixed values of $k$, what is the asymptotic number $s_{n,\lfloor (n-1)/2
\rfloor - k}$ of saturated secondary structures having exactly
$k$ base pairs fewer than the maximum?
\end{openproblem}
Note that in \cite{Clote.jcb06}, we solved this problem for $k=0,1$.

A related interesting question concerns whether the number of secondary
structures $s_{n,k}$ having $k$ base pairs is approximately Gaussian.
As first suggested by Y. Ponty (personal communication), this is indeed
the case. More formally, consider for fixed $n$ the 
the finite distribution $\mathbb{P}_n = p_1,\ldots,p_n$, where
$p_k = s_{n,k}/s_n$ and $s_n = \sum_k s_{n,k}$. In the Nussinov energy
model, the energy of a secondary structure with $k$ base pairs is $-k$,
so the distribution $\mathbb{P}_n$ is what is usually called the
{\em density of states} in physical chemistry. It follows from Theorem 1 of
of Drmota \cite{drmota} (see also \cite{drmota1994})
that $\mathbb{P}_n$ is Gaussian. Similarly, it
follows from Theorem 1 of Drmota that the asymptotic distribution of
density of states of both canonical and saturated structures is Gaussian.
Details of a Maple session applying Drmota's theorem to saturated structures
appears in the web supplement
\url{http://bioinformatics.bc.edu/clotelab/SUPPLEMENTS/JBCBasymptotics/}.

\subsection{Asymptotic number of saturated stem-loops}
\label{section:saturatedStemLoops}

Define a {\em stem-loop} to be a secondary structure $S$
having a unique base pair $(i_0,j_0)\in S$, for which all other
base pairs $(i,j) \in S$ satisfy the relation $i < i_0 < j_0 < j$.
In this case, $(i_0,j_0)$ defines a hairpin, and the remaining base pairs,
as well as possible internal loops and bulges, constitute the stem.
We have the following simple result due to Stein and Waterman
\cite{steinWaterman}.
\begin{proposition}
There are $2^{n-2}-1$ stem-loop structures\footnote{In \cite{steinWaterman},
stem-loop structures are called {\em hairpins}. Since the appearance of
\cite{steinWaterman}, common convention is that a hairpin is a structure
consisting of a single base pair enclosing a loop region; i.e.
$\op \ub \cdots \ub \cp$. Here we use the more proper term {\em stem-loop}.}
on $[1,n]$.
\end{proposition}
\begin{proof}{}
Let $L(n)$ denote the number
of secondary structures with {\em at most} one
loop on $(1,\ldots,n)$. Then
$L(1) = 1 = L(2)$.
There are two cases to consider for $L(n+1)$.
\medskip

\noindent
{\sc Case 1.}
If $n+1$ does not form a base pair,
then we have a contribution of $L(n)$.
\medskip

\noindent
{\sc Case 2.}
$n+1$ forms a base pair with some $1 \leq j \leq n-1$.
In this case, since only one hairpin loop is allowed,
there is no base-pairing for the subsequence
$s_1,\ldots,s_{j-1}$, and hence if $n+1$ base-pairs with $j$, then
we have a contribution of
$L(n-(j+1)+1) = L(n-j)$.
Hence
\begin{eqnarray*}
L(n+1) & = & L(n) + \sum\limits_{j=1}^{n-1} \, L(n-j)\\
       & = & L(n) + L(n-1)+ \cdots + L(1)
\end{eqnarray*}
and hence
$L(1)=1$, $L(2)=1$, $L(3)=2$, 
and from there $L(n) = 2^{n-2}$ by induction.
\end{proof}
We now compute the asymptotic number of {\em saturated} stem-loop structures.
Let $h(n)$ be the number of saturated stem-loops on $[1,n]$, defined by
$h(n)=1$ for $n=0,1,2,3$, $h(4)=3$,
and
\begin{eqnarray}
\label{eqn:numSaturatedStemLoops}
h(n) = h(n-2)+2 h(n-3)+2 h(n-4)
\end{eqnarray}
for $n \geq 5$. Note that we have defined $h(1)=1=h(2)$ for notational
ease in the sequel, although there are in fact no stem-loops of size $1$
or $2$. Indeed in this case, the only structures of size $1$ respectively $2$ are
$\ub$ and $\ub \ub$.

The first few terms in the sequence
$h(1), h(2), h(3), \cdots$ are 
$1$, $1$, $1$, $3$, $5$, $7$, $13$,
$23$, $37$, $63$, $109$, $183$, $309$,
$527$, $893$, $1511$, $2565$, $4351$, $7373$, $12503$;
for instance, $h(20)=12503$.
\subsubsection*{Grammar}
It is easily seen that the following rules
\[
S \rightarrow \ub | \ub \ub | 
\op S \cp | \ub \op S \cp | \ub \ub \op S \cp | \op S \cp \ub | 
\op S \cp \ub \ub 
\]
provide for a non-ambiguous context-free grammar to generate all non-empty
saturated stem-loops. It defines actually a special kind of context-free language, called regular, whose generating function is rational.  
\subsubsection*{Generating Function}
By the DSV methodology, we obtain the functional relation
\[
R(z) = z + z^2 +R(z) z^2 + 2 R(z) z^3 + 2 R(z) z^4
\]
whose solution is the rational function
\begin{eqnarray}
\label{eqn:rationalFunctionNumSatStemLoops}
R(z) = \frac{P(z)}{Q(z)} = \frac{z}{1-z-2 z^3}
\end{eqnarray}
where $P(z)=z$ and $Q(z)=1-z -2 z^3$.
\subsubsection*{Asymptotics}
For rational functions, an easy way to compute the asymptotic behaviour of the Taylor coefficients is to compute a partial fraction decomposition and isolate the dominant part. This is equivalent to solving the corresponding linear recurrence. See also~\cite[p.~325]{grahamKnuthPatashnik} or \cite[Thm. 9.2]{odlyzko:HandbookOfCombinatorics}.

Partial fraction decomposition yields
\[R(z)=\frac{A(a_1)}{1-z/a_1}+\frac{A(a_2)}{1-z/a_2}+\frac{A(a_3)}{1-z/a_3},\]
where the~$a_i$s are the roots of~$Q$ and $A(z)=-1/Q'(z)$. It follows by extracting coefficients that
\[h(n)=A(a_1)a_1^{-n}+A(a_2)a_2^{-n}+A(a_3)a_3^{-n}.\]
(Note that this is an actual equality valid for all $n\ge0$ and not an asymptotic result).
Now, the roots of~$Q$ are approximately
\[a_1=0.5897545,\quad a_2=-0.294877 - 0.872272 i,\quad a_3=-0.294877 + 0.872272 i.\] 
Since $|a_2|=|a_3|=.9207>|a_1|$, it follows that the asymptotic behaviour is given by the term in~$a_1$. 

We have proved the following theorem.
\begin{theorem}
\label{thm:numberCanonicalStemLoops}
The number $h(n)$ of saturated stem-loops on
$[1,n]$ satisfies
\begin{eqnarray}
\label{eqn:closedFormNumSatStemLoops}
h(n) \sim  0.323954 \cdot 1.69562^n.
\end{eqnarray}
\end{theorem}
Convergence of the asymptotic limit
in equation (\ref{eqn:closedFormNumSatStemLoops}) is
exponentially fast, so that when $n=20$,
$0.323954 \cdot 1.69562^n=12504.2$, while the exact number of
saturated stem-loops on $[1,20]$ is $h(20)=12503$.

\section{Quasi-random saturated structures}
\label{qUniform distribution}
\newcommand{\qsat}{quasi saturated }

In this section, we define a stochastic
greedy process to generate {\em random} saturated structures, 
technically denoted {\em quasi-random saturated structures}.
Our main result is that the expected number of base pairs in
quasi-random saturated structures is $0.0.340633 \cdot n$, just slightly more 
than the expected number $0.337361 \cdot n$ of all saturated structures.
This suggests that the introduction of stochastic greedy algorithms and
their asymptotic analysis may prove useful in other areas of
random graph theory.

Consider the following stochastic process to generate a saturated 
structure. Suppose that $n$ bases are arranged in sequential order on a line.
Select the base pair $(1, u)$ by choosing
$u$, where $\theta +2 \leq u \leq n$, at random with  
probability $1 / (n-\theta -1)$. 
\begin{figure}[htb!]
\begin{center}
\includegraphics[width=10cm]{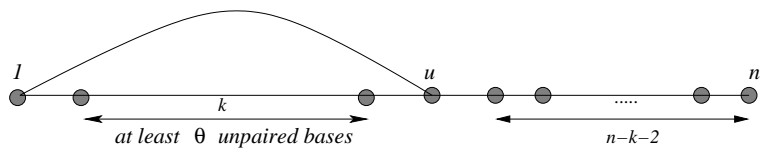}
\caption{Base $1$ is base-paired by selecting a random base $u$ such there
are at least $\theta$ unpaired bases enclosed between $1$ and $u$.}
\label{line}
\end{center}
\end{figure}
The base pair
joining $1$ and $u$ partitions the line into two parts. The left region
has $k$ bases strictly between $1$ and $u$, where
$k \geq \theta $, and the right region contains the remaining
$n-k-2$ bases properly contained within endpoints $k+2$ and $n$ 
(see Figure~\ref{line}). Proceed recursively on each of the two parts.
Observe that the secondary structures produced
by our stochastic process will always base pair with the leftmost
available base, and that the resulting structure is always saturated.

Before proceeding further, we note that the probability that the
probability $p_{i,j}$ that $(i,j)$ 
is a base pair in a saturated structure is {\em not} the same
as the probability $q_{i,j}$ that $(i,j)$ 
is a base pair in a quasi-random saturated structure.  Indeed, if
we consider saturated and quasi-random saturated structures on an
RNA sequence of length $n=10$, then clearly $p_{1,5}=1/29$ while
clearly $q_{1,5} = 1/8$.\footnote{The web supplement contains a Python
program to compute the number of saturated structures on $n$. Clearly
$p_{1,5} = \frac{s_3 \cdot s_5}{s_{10}}$, where $s_k$ denotes the number
of saturated structures on an RNA sequence of length $k$. A computation
from a Python program (see web supplement) shows that $s_3=1$, $s_5=5$ and
$s_{10}=145$, hence $p_{1,5} = 5/145 = 1/29$.} Despite the very different
base pairing probabilities when comparing saturated with quasi-random
saturated structures, it is remarkable that the expected number of base
pairs over saturated and quasi-random saturated structures is numerically
so close.

Let $U^{\theta}_n$ be the expected number of base pairs of the 
saturated secondary
structure generated by this recursive procedure.
In general, we have the following recursive equation 
\begin{eqnarray}
U^{\theta}_n  
&=& \label{qeq3b}
1 + \frac{1}{n-\theta -1} \sum_{k=\theta}^{n-2} 
(U^{\theta}_k +U^{\theta}_{n-k-2}),\qquad n\ge\theta+2,
\end{eqnarray}
with initial conditions
\begin{equation}\label{eq:inicond}
 U^{\theta}_0 =  U^{\theta}_1 = \cdots =  U^{\theta}_{\theta+1} = 0,\quad
U^{\theta}_{\theta+2} = U^{\theta}_{\theta+3} = 1.
\end{equation}
If we write
equation (\ref{qeq3b}) for $ U^{\theta}_{n+1}$ and substitute in it the value for $ U^{\theta}_n$ 
we derive
\begin{eqnarray*}
 U^{\theta}_{n+1} 
& = & 1 +  \frac{1}{n-\theta} \sum_{k=\theta}^{n-1}  (U^{\theta}_k +U^{\theta}_{n-k-1})\\
& = & 1 +  \frac{1}{n-\theta} \left( U^{\theta}_{n-1} + U^{\theta}_{n-\theta-1} +
\sum_{k=\theta}^{n-2}  (U^{\theta}_k +U^{\theta}_{n-k-2}) \right)\\
& = & 1 +  \frac{1}{n-\theta} \left( U^{\theta}_{n-1} + U^{\theta}_{n-\theta-1} \right) +
\frac{n-\theta-1}{n-\theta} (U^{\theta}_n -1 )  .
\end{eqnarray*}
If we multiply out by $n-\theta$ and simplify we obtain
\begin{equation}
\label{qeq4b}
(n-\theta)  U^{\theta}_{n+1} = 1 + (n-\theta-1)  U^{\theta}_n + U^{\theta}_{n-1} + U^{\theta}_{n-\theta-1},
\end{equation}
which is valid for $n \geq \theta +1$.

\subsection{Asymptotic behavior}

We now look at asymptotics. In particular we prove
the following result.
\begin{theorem}
\label{qgen:thm}
Let $U_n^\theta$ denote the expected number of base pairs for quasi-random
saturated structures of an RNA sequence of length $n$. 
Then for fixed~$\theta$, and as $n\rightarrow\infty$
\begin{equation}
\label{qeq:gen} 
U_n^\theta \sim K_\theta\cdot n\qquad\text{with}\qquad K_\theta=e^{-1-H_{\theta+1}}\int_0^1{e^{t+(t+t^2/2+\dots+t^{\theta+1}/(\theta+1))}\,dt},
\end{equation}
where 
$H_{\theta+1}=1+\frac{1}{2}+\dots+\frac{1}{\theta+1}$ is the $(\theta+1)$th harmonic number.
\end{theorem}
The first few values can easily be obtained numerically and we have
\[K_1=0.340633,\quad K_2=0.285497,\quad K_3=0.247908,\quad K_4=0.220308,\quad K_5=0.199018.\]
\begin{proof}
For fixed integer~$\theta$, the recurrence~\eqref{qeq4b} is linear with polynomial coefficients. It is a classical result that the generating functions of solutions of such recurrences satisfy linear differential equations.  This is obtained by applying the following rules: if $U(z)=\sum_{n\ge0}u_nz^n$, then
\[\sum_{n\ge0}nu_nz^n=zU'(z),\qquad \sum_{n\ge 0}u_{n+k}z^n=\frac{1}{z^k}(U(z)-u_0-u_1z-\dots-u_{k-1}z^{k-1}).\]
Starting from~\eqref{qeq4b}, we first shift the index by~$\theta+1$ and apply these rules together with the initial conditions~\eqref{eq:inicond} to get
\begin{align*}
	(n+\theta+2)U_{n+\theta+2}^\theta-(\theta+1)U_{n+\theta+2}^\theta&=1+(n+\theta+1)U_{n+\theta+1}^\theta-(\theta+1)U_{n+\theta+1}^\theta+U_{n+\theta}^\theta+U_{n}^\theta,\\
	\frac{1}{z^{\theta+2}}zy'-(\theta+1)\frac{y}{z^{\theta+2}}&=\frac{1}{1-z}+\frac{1}{z^{\theta+1}}zy'-(\theta+1)\frac{y}{z^{\theta+1}}+\frac{y}{z^\theta}+y.
\end{align*}
Finally, this simplifies to

\begin{equation}
\label{qeq7b}
z(1-z)y' +((\theta+1)(z-1)-z^{2}-z^{\theta+2}) y =\frac{z^{\theta+2}}{1-z}.
\end{equation}
This is a first order non-homogeneous linear differential equation. The homogeneous part 
\[z(1-z)W' +((\theta+1)(z-1)-z^{2}-z^{\theta+2}) W =0\]
is solved by integrating a partial fraction decomposition 
\begin{align*}
\frac{W'(z)}{W(z)}&=\frac{\theta+1}{z}-\frac{z}{z-1}-\frac{z^{\theta+1}}{z-1}\\
&=\frac{\theta+1}{z}+\frac{2}{z-1}-1-(1+z+\dots+z^{\theta})\\
\log W&=(\theta+1)\log z-2\log(1-z)-z-(z+z^2/2+\dots+z^{\theta+1}/(\theta+1)),\\
W(z)&=\frac{z^{\theta+1}}{(1-z)^2}e^{-z-(z+z^2/2+\dots+z^{\theta+1}/(\theta+1))}.
\end{align*}
From there, variation of the constant gives the following expression for the generating function:
\[y=\frac{z^{\theta+1}}{(1-z)^2}e^{-z-(z+z^2/2+\dots+z^{\theta+1}/(\theta+1))}\int_0^z{e^{t+(t+t^2/2+\dots+t^{\theta+1}/(\theta+1)}\,dt}.\]
Because the exponential is an entire function, we readily find that the only singularity is at~$z=1$, where
$y\sim {K}/{(1-z)^2}$ with $K$ as in the statement of the theorem.
%
The proof is completed by the use of  
Theorem~\ref{thm:flajolet}.
\end{proof}

Note that the
asymptotic expected number of base pairs in quasi-random saturated structures with $\theta=1$
is $0.340633 \cdot n$, while by
Theorem \ref{thm:expectedNumBasePairsInSatStr} the
asymptotic expected number of base pairs in saturated structures is
$0.337361 \cdot n$, just very slightly less. This result
points out that the stochastic greedy method performs reasonably well
in sampling saturated structures, although the stochastic process
tends not to sample certain (rare)
saturated structures having a less than average number of
base pairs.

The stochastic process used to construct quasi-random saturated
structures iteratively base-pairs the leftmost position in each subinterval.
One can imaging a more general stochastic method of constructing
saturated structures, described as follows.
Generate an initial list $L$ of all allowable base pairs $(i,j)$
with $1 \leq i < j \leq n$ and $j \geq i+\theta+1$.
Create a saturated structure by repeately picking a base pair from 
$L$, adding it to an initially empty structure $S$,
then removing from $L$ all base pairs that form a crossing 
(pseudoknot) with the base pair just selected. This ensures
that the next time a base pair from $L$, it can be added to $S$
without violating the definition of secondary structure.
Iterate this procedure until $L$ is empty to form the stochastic saturated
structure $S$.

Taking an average over 100 repetitions, we have computed the
average number of base pairs and the standard deviation for
$n =10, 100, 1000$. Results are $\mu = 0.323$, $\sigma= 0.0604$  for $n=10$, 
$\mu = 0.3526$, $\sigma= 0.0386$ for $n=100$
and $\mu = 0.35618$, $\sigma= 0.0361$ for $n=1000$.
This clearly is a different stochastic process than that used for
quasi-random saturated structures.

\section{Conclusion}\label{conclusion}

In this paper we applied the DSV methodology and the Flajolet-Odlyzko
theorem to asymptotic enumeration problems
concerning  canonical and saturated secondary structures.
For instance, we showed that
the expected number of base pairs in 
canonical RNA secondary structures is equal to $0.31724 \cdot n$,
which is far less than the expected number 
$0.495917 \cdot n$ of base pairs over all secondary structures,
the latter which follows from Theorem 4.19 of \cite{hofacker98}.
This may provide a theoretical
explanation for the speed-up observed for Vienna RNA Package when restricted
to canonical structures \cite{Bompfunewerer.jmb08}. 

Additionally, we computed the asymptotic number 
$1.07427\cdot n^{-3/2} \cdot 2.35467^n$ of
saturated structures, the expected number $0.337361 \cdot n$
of base pairs of saturated
structures and the asymptotic number 
$0.323954 \cdot 1.69562^n$ of saturated stem-loop structures.
We then considered a natural stochastic greedy process to generate
quasi-random saturated structures, and 
showed surprisingly that the expected number
of base pairs of is $0.340633 \cdot n$, a value very close to the
expected number $0.337361 \cdot n$ of base pairs of all saturated
structures. Finally, we apply a theorem of Drmota \cite{drmota} to
show that the density of states for [all resp. canonical resp. 
saturated] secondary structures is asymptotically Gaussian.

\section*{Acknowledgements}
We would like to thank Yann Ponty, for suggesting that Drmota's work
can be used to prove that the density of states for secondary structures
is Gaussian. Thanks as well to two anonymous referees, whose comments
led to important improvements in this paper.
Figure \ref{fig:FlajoletTriangle} is due to W.A. Lorenz, and first 
appeared in the joint article Lorenz et al. \cite{lorenzPontyClote:asymptotics}.

Funding for the research of 
P. Clote was generously provided by the
Foundation Digiteo - Triangle de la
Physique and the National Science Foundation
DBI-0543506 and DMS-0817971. Additional support is
gratefully acknowledged to the
Deutscher Akademischer Austauschdienst for a visit to
Martin Vingron's group in the Max Planck Institute of
Molecular Genetics.
Any opinions, findings,
and conclusions or recommendations expressed in this material are
those of the authors and do not necessarily reflect the views of the
National Science Foundation.
Funding for the research of 
E. Kranakis was generously provided by the
Natural Sciences and Engineering Research Council
of Canada (NSERC) and Mathematics of Information Technology and Complex
Systems (MITACS).
Funding for the research of B. Salvy was provided by
Microsoft Research-Inria Joint Centre.

\bibliographystyle{plain}

\end{document}